\documentclass[letterpaper, 10 pt, conference]{ieeeconf}
\usepackage{kjkslim}
\usepackage{comment} % to comment out blocks of code

\IEEEoverridecommandlockouts                              
\makeatletter
\let\NAT@parse\undefined
\makeatother
\usepackage[]{cite}

\usepackage{xcolor}
\definecolor{linkblue}{RGB}{0,0,150}

\usepackage[
  colorlinks=true,
  linkcolor=red,
  citecolor=blue,
  urlcolor=linkblue
]{hyperref}

\usepackage{multirow}

\title{\LARGE \bf
Scalable Supervisory HVAC Control for Linear Objectives
}

\author{W. Grant Dierking, Arash J. Khabbazi, Levi D. Reyes Premer, Kevin J. Kircher, \IEEEmembership{Member,~IEEE}
\thanks{School of Mechanical Engineering, Purdue University, West Lafayette, Indiana, USA. \{{\tt \small wdierkin, arashjkh, lreyespr, 
kkirche}\}{\tt \small@purdue.edu}}
}

\begin{document}

\maketitle
\thispagestyle{empty}
\pagestyle{empty}

\begin{comment}
\begin{abstract}
Advanced control of heating and cooling systems can substantially reduce energy costs and pollution. However, real-world adoption of popular algorithms among researchers, such as model predictive control (MPC) and reinforcement learning (RL), remains limited due in part to their high deployment and commissioning costs. Here, we develop two nearly commissioning-free controllers tailored to objectives that depend linearly on the controlled thermal load, such as energy costs and pollution. The controllers require at most two thermal parameters. In representative cooling simulations, controller performance is largely insensitive to parameter values, suggesting potential for deployment with little to no tuning. With a time-invariant electricity price, both controllers achieve near The controllers maintain occupant comfort while achieving 23.9--100\% (depending on the electricity pricing and controller variant) of the performance improvement achieved by an omniscient policy with perfect model information and forecasts. These results suggest that simple, structure-exploiting controllers may capture most of the attainable value of advanced control while avoiding the data, modeling, tuning, and computational burdens that can arise with conventional MPC or RL. 
\end{abstract}
\end{comment}

\begin{abstract}
Advanced control of heating, ventilation, and air-conditioning (HVAC) systems can substantially reduce energy costs and pollution. However, real-world adoption of popular algorithms among researchers, such as model predictive control and reinforcement learning, remains limited due in part to their high deployment and commissioning costs. Here, we develop two nearly commissioning-free supervisory controllers tailored to objectives that depend linearly on the controlled thermal load, such as energy costs and pollution. The controllers require, at most, two easily-estimable thermal parameters, forecasts of energy prices and occupant temperature preferences over a prediction horizon, and an indoor temperature measurement. In residential cooling simulations, both controllers perform essentially as well under traditional time-invariant electricity pricing as an omniscient optimal controller with exact model information and perfect forecasts, and attain up to 86.6\% of the omniscient cost savings under increasingly prevalent time-varying pricing. These results suggest that simple, structure-exploiting controllers may capture most of the attainable value of advanced supervisory HVAC control with linear objectives, while avoiding the data, modeling, tuning, and computational burdens that hinder real‑world deployment. 
\end{abstract}

\section{Unrealized Potential of HVAC Control}
\label{intro}

Heating, ventilation, and air‑conditioning (HVAC) systems in residential and commercial buildings shape occupant comfort and health while causing annual energy costs on the order of \$1 trillion and about 15\% of global greenhouse‑gas emissions \cite{Khabbazi-2025-Field-HVAC-MPC-RL}. Decades of control research have produced a rich set of algorithms, including model predictive control (MPC; see \cite{Serale2018-ph, Drgona2020-fh}), data‑enabled predictive control (DeePC; see \cite{Lian2021, Chinde2022}), and reinforcement learning (RL; see \cite{Vazquez-Canteli2019-ey, Wang2020-lx}), which could significantly reduce energy costs, pollutant emissions, and peak electricity demand. A recent review of field experiments showed that many advanced HVAC controllers produce nontrivial real‑world savings in residential and commercial buildings when deployed with objectives that are linear in the controlled thermal load \cite{Khabbazi-2025-Field-HVAC-MPC-RL}.

Despite these promising results, industry adoption of advanced HVAC control remains limited, in part because the cost and complexity of commissioning are recurrent barriers \cite{Henze-etal-2025}. High-performance MPC requires detailed plant models and careful tuning \cite{Blum2019}; DeePC and RL typically require large, high‑quality datasets for training; and safe RL deployment often still relies on accurate simulation models or conservative exploration strategies \cite{Mulayim2025}. Practitioners and vendors repeatedly cite these commissioning and integration burdens as primary obstacles to scaling advanced control in commercial and residential markets \cite{Henze-etal-2025, Pergantis-etal-2024a}.

This paper investigates the level of complexity that is actually necessary to capture the majority of attainable value for supervisory objectives that are linear in thermal load. We show that under first‑order thermal dynamics, supervisory setpoint optimization reduces to a linear program in zone temperatures that depends only on, at most, two easily-estimable parameters, an indoor temperature measurement, occupant comfort preferences, and prices. This result removes the need to identify thermal resistances or disturbance distributions and exposes a simple structure that supervisory controllers can exploit.

Building on that structure, we introduce two nearly commissioning‑free supervisory controllers. The zeroth‑order controller implements schedule‑aware setpoint shifts with temperature ramp-rate limits to minimize costs during unoccupied periods. The first‑order predictive controller additionally uses a time constant to shift thermal loads and better respond to dynamic occupant comfort preferences. Both parameters can be set from intuition, a short test, or rules of thumb. In representative residential simulations under popular electricity pricing structures, these controllers recover significant attainable energy cost savings while imposing a small fraction of the commissioning burden typical of MPC or data‑driven methods. We quantify sensitivity to parameter values, identify the principal limitations, and outline directions for field validation and extension.

Section \ref{problem} formalizes the supervisory HVAC control problem and the simplifying assumptions used throughout the paper. Section \ref{design} presents the proposed zeroth‑ and first‑order controllers. Section \ref{numerical} discusses their performance in simulation, and Section \ref{conclusions} concludes the paper.

\section{Problem Statement}
\label{problem}

We consider a residential or commercial building, or a thermal zone within a larger building, and index discrete time steps of duration $\Delta t$ (h) by the integer $k \in \setof{0, \dots, K}$. One representative indoor temperature measurement $y(k)$ ($^\circ$C) is available at each time. The thermal dynamics are unknown. A well-tuned low-level controller adjusts equipment operating points, such as valve positions or compressor and fan speeds, so that the measured temperature $y(k+1)$ tracks a setpoint $u(k)$ ($^\circ$C) chosen by the supervisory controller. Forecasts of weather, energy prices, and occupancy are available.

The supervisory problem is to choose setpoints $u(k)$ that minimize
\begin{equation}
    \mathbb{E}\!\left[ \Delta t \sum_{k=0}^{K-1} \pi(k)\, q(k) \right], \label{objective}
\end{equation}
where $\pi(k)$ (\$/kWh) is an effective thermal price and the thermal load $q(k)$ (kW) met by the HVAC equipment is positive for heating and negative for cooling. Several objectives of practical interest are linear in $q(k)$. For example, energy costs can be minimized by setting $\pi(k) = \pm \pi^\text{fuel}(k)/\eta(k),$ where $\pi^\text{fuel}(k)$ is the fuel or electricity price, $\eta(k) > 0$ is the equipment efficiency or coefficient of performance (COP), and the plus and minus signs correspond to heating and cooling, respectively. Pollution can be minimized by setting $\pi(k)= \pm \mu(k)/\eta(k)$, where $\mu(k)$ (kg/kWh) is the pollutant intensity. Setting $\pi(k) = \pm(\pi^\text{fuel}(k) + \pi^\text{pol}(k) \mu(k))/\eta(k)$, where $\pi^\text{pol}(k)$ (\$/kg) is a pollutant price, balances minimization of energy costs and pollution.

The temperature measurements should satisfy
\begin{equation}
\begin{aligned}
    y^{\min}(k) &\le y(k) \le y^{\max}(k) \\
    r^{\min}(k) &\le \frac{y(k)-y(k-1)}{\Delta t} \le r^{\max}(k)
\end{aligned} \label{constraints}
\end{equation}
for $k = 1,\dots,K$. The temperature bounds enforce occupant comfort preferences, while the ramp-rate limits prevent abrupt temperature swings and help keep equipment within capacity limits. The low-level controller saturates equipment as needed to respect physical limits. The unknown relationships between $u(k)$, $q(k)$, and $y(k)$ may include nonlinearities, unobserved states, disturbances, or measurement noise.

Our overarching goal in this paper is to design a supervisory controller that achieves good performance with minimal commissioning burden, which we quantify by the number and complexity of system parameters that must be identified before effective control can begin. The next section formalizes the supervisory control problem, exposing structure that enables nearly commissioning-free controllers.

\section{Control Design}
\label{design}

This section develops two controllers, which we refer to as zeroth- and first-order, that both formulate supervisory control as a linear program in the indoor temperature only:
\bneq
\begin{aligned}
\text{minimize} \quad &\sum_{k=1}^K c(k) T(k) \\
\text{subject to} \quad &\eqref{constraints} \text{ for } k = 1, \dots, K \label{reformulatedProblem} \\
\end{aligned}
\eneq
with variables $T(1)$, \dots, $T(K)$. The input data are the initial temperature $T(0)$ and, for $k = 1$, \dots, $K$, $c(k)$, $y^\text{min}(k)$, $y^\text{max}(k)$, $r^\text{min}(k)$, and $r^\text{max}(k)$. The zeroth- and first-order controllers will be characterized by their distinct choices of the parameters $c(k)$ in the objective function, which will follow from different underlying assumptions.

The striking feature of problem \eqref{reformulatedProblem} is that it requires very little information about the controlled system. The problem makes no reference to any system models or to important but difficult-to-estimate disturbances, such as heat transfer from sunlight and occupant activity.

%This reformulation eliminates the variables $q(0)$, \dots, $q(K-1)$, leaving only $T(1)$, \dots, $T(K)$. The reformulation also eliminates the need to specify $R$, an influential parameter that can be challenging to estimate, and requires no distributional information other than the expected prices $\bar \pi(k)$. In particular, Problem \eqref{prob2} requires no explicit information about the disturbances $w(k)$, which depend through the $q^\text{exog}(k)$ on effects that can be challenging to model, such as solar heat gains, occupant activity, and nonlinearities or higher-order dynamics that the first-order linear model \eqref{1r1c} neglects. These effects are all aggregated into a single set of ramp-rate limits $r^{min}$ and $r^{max}$, which are estimable through human intuition rather than data-regression. The only other model information that Problem \eqref{prob2} requires is the time constant $\tau$, which determines the parameter $a$ used in \eqref{price1}. 

\subsection{Motivating Physical Intuition}
\label{physics}

To motivate control design, we consider perhaps the simplest useful model of a thermal system. The model has the form of a thermal circuit with one resistor and one capacitor. In the thermal circuit analogy, temperature plays the role of voltage and heat plays the role of charge \cite{Khabbazi2026-ACC}. While this simple model guides control design, the resulting controllers use few or no model parameters and can be applied to much more complex thermal systems. For both controllers, we assume perfect temperature measurements and one-step-ahead setpoint tracking: $y(k+1) = T(k+1) = u(k)$.

In continuous time $t$ (h), the model is
\bneq
C \dot T(t) = \frac{\theta(t) - T(t)}{R} + q(t) + q^\text{exog}(t) , \label{continuousDynamics}
\eneq
where $T$ ($^\circ$C) is the indoor air temperature, $\theta$ ($^\circ$C) is a boundary temperature (which could be the outdoor temperature or a mixture of the outdoor temperature and adjacent zone temperatures), $R$ ($^\circ$C/kW) is the thermal resistance between $T$ and $\theta$, $C$ (kWh/$^\circ$C) is the thermal capacitance, $q$ (kW) is the thermal load met by HVAC equipment, and $q^\text{exog}$ (kW) is the thermal power from sunlight, body heat, lights, plug loads, and other exogenous sources (possibly including discrepancies between the model and the true physics).

\subsection{Zeroth-Order Problem}
\label{problem0}

To derive the zeroth-order problem, we assume that the thermal capacitance $C$ is negligible, meaning the left-hand side of \eqref{continuousDynamics} is identically zero and the indoor temperature $T$ responds instantaneously to changing boundary conditions. This implies that $T(k) = R ( q(k) +  w(k) )$ for all $k$, where
\bneq
%q(k) = \frac{T(k)}{R} - w(k) , \text{ where } w(k) = \frac{ \theta(k) }{ R } + q^\text{exog}(k) . \label{qss}
w(k) = \theta(k) / R  + q^\text{exog}(k) . \label{qss}
\eneq
Therefore, the zeroth-order supervisory control problem is to
\bneq
\begin{aligned}
&\text{minimize} \quad  \mathbb{E}\!\left[ \Delta t \sum_{k=0}^{K-1} \pi(k)\, q(k) \right] \\
&\text{subject to} \\
&\quad T(k) = R ( q(k) +  w(k) ) \text{ for } k = 1, \dots, K \\
&\quad \eqref{constraints} \text{ for } k = 1, \dots, K \label{prob0} \\
\end{aligned}
\eneq
with variables $q(0)$, \dots, $q(K-1)$ and $T(1)$, \dots, $T(K)$.  Proposition \ref{prop0} states that \eqref{prob0} is equivalent to \eqref{reformulatedProblem} if 
\bneq
\begin{aligned}
c(k) &= \bar \pi(k), \ k = 1, \dots, K, \label{price0} \\
\end{aligned}
%\begin{aligned}
%c(k) &= \bar \pi(k-1), \ k = 1, \dots, K - 1 \\
%c(K) &= 0 , \label{price0} \\
%\end{aligned}
\eneq
where $\bar \pi(k) = \Expect \pi(k)$. 

\begin{proposition} \label{prop0}
If $c(1)$, \dots, $c(K)$ are defined as in \eqref{price0}, then a sequence of setpoints $T^\star(1)$, \dots, $T^\star(K)$ is optimal for Problem \eqref{prob0} if and only if it is optimal for Problem \eqref{reformulatedProblem}.
\end{proposition}

\begin{proof}
Using the equality constraint in \eqref{prob0}, we eliminate $q(k) = T(k)/R - w(k)$ from the objective. We then multiply the objective by $R / \Delta t > 0$ and subtract additive constants involving the initial temperature $T(0)$ and the disturbances $w(k)$. Eliminating equality constraints, multiplying the objective by a positive scalar, and subtracting additive constants from the objective are standard operations in conversions between equivalent problem formulations \cite{Boyd2004}.
\end{proof}

\subsection{First-Order Problem}
\label{problem1}

To derive the first-order problem, we assume that $\theta$, $q$, and $q^\text{exog}$ are piecewise constant over each time step. This implies that the dynamics \eqref{continuousDynamics} discretize exactly to
\bneq
T(k+1) = a T(k) + (1-a) R ( q(k) + w(k) ) , \label{1r1c}
\eneq
where $a = \exp(-\Delta t / \tau)$, $\tau = R C$ (h) is the time constant and $w(k)$ is defined as in \eqref{qss}. With the dynamics \eqref{1r1c}, %and assuming perfect temperature measurements and setpoint tracking ($y(k) = T(k) = u(k)$), 
the first-order supervisory control problem is to
\bneq
\begin{aligned}
&\text{minimize} \quad \Expect \left[ \Delta t \sum_{k=0}^{K-1} \pi(k) q(k) \right] \\
&\text{subject to} \\
&\quad T(k) = a T(k-1) + (1-a) R ( q(k-1) \\
&\qquad + w(k-1) ) \text{ for } k = 1, \dots, K  \label{prob1} \\
&\quad \eqref{constraints} \text{ for } k = 1, \dots, K  \\
\end{aligned}
\eneq
with variables $q(0)$, \dots, $q(K-1)$ and $T(1)$, \dots, $T(K)$. %The input data are $\Delta t$, $T(0)$, $a$, and $R$; $y^\text{min}(k)$, $y^\text{max}(k)$, $r^\text{min}(k)$, and $r^\text{max}(k)$ for $k = 1$, \dots, $K$; and the joint distribution of the random variables $\pi(0)$, \dots, $\pi(K-1)$, $w(0)$, \dots, $w(K-1)$.
Proposition \ref{lem1} states that \eqref{prob1} is equivalent to \eqref{reformulatedProblem} if
\bneq
\begin{aligned}
c(k) &= \bar \pi(k-1) - a \bar \pi(k) , \  k = 1, \dots, K - 1 \\
c(K) &= \bar \pi(K-1) . \label{price1} \\
\end{aligned}
\eneq

%The remaining quantities (temperature limits and ramp-rate limits) are user‑specified comfort and feasibility constraints. Although these limits may be chosen conservatively or estimated from typical HVAC capabilities, they are not required for the model reduction and do not constitute physical model parameters. 

%While these parameters can be identified with high precision using data regression techniques, we have intentionally chosen to reformulate the problem using these parameters under the assumption that they can be easily estimated as a constant value within an acceptable margin of error during deployment. In particular, we assume it reasonable for an installer with limited knowledge of controls and building thermal dynamics to be able to estimate these parameters when provided with guidelines of typical values for different building and HVAC system types. Performance sensitivity to parameter estimation error is evaluated for a numerical example in Section ~\ref{sensitivity}.

%In the numerical experiments in \S\ref{experiments}, controller performance is robust to large errors in specifying $\tau$ and moderate errors in specifying rate-of-change constraints. 

%Most buildings have time constants of one to 10 hours. 

\begin{proposition} \label{lem1}
If $c(1)$, \dots, $c(K)$ are defined as in \eqref{price1}, then a sequence of setpoints $T^\star(1)$, \dots, $T^\star(K)$ is optimal for Problem \eqref{prob1} if and only if it is optimal for Problem \eqref{reformulatedProblem}.
\end{proposition}

\begin{proof}
The proof resembles the proof of Proposition \ref{prop0}. Using the equality constraint in \eqref{prob1}, we eliminate
\[
q(k) = \frac{T(k+1) - a T(k)}{(1 - a) R} - w(k) ,
\]
multiply the objective \eqref{objective} by 
the constant $(1-a) R / \Delta t$ (which is positive because $a < 1$), and subtract additive constants involving $T(0)$ and the $w(k)$. This yields the equivalent objective 
\bneq
\sum_{k=0}^{K-1} \bar \pi(k) ( T(k+1) - a T(k) ) . \label{newObjective}
\eneq
The $c(k)$ definitions \eqref{price1} come from expanding the sum in \eqref{newObjective} and grouping the terms multiplying each $T(k)$.
\end{proof}

\subsection{Predictive Control Implementation}
\label{algorithm}

In some contexts, it may make sense to solve \eqref{reformulatedProblem} in an open-loop fashion. For example, \eqref{reformulatedProblem} could be solved once per day to generate a temperature setpoint schedule that a human operator could screen before implementation. This approach is likely to perform well when price forecasts are updated infrequently, measurement noise is small, and the low-level control system tracks setpoints well. 

In contexts with significant uncertainty from prices, noise, imperfect tracking, or other sources, performance can likely be improved by periodically re-optimizing in a receding-horizon fashion.
%Strict temperature constraints may occasionally render the problem infeasible due to small modeling inaccuracies or unexpected disturbances. A simple fallback policy, designating the nearest comfortable temperature as the setpoint for the next time step regardless of ramp-rate constraints, ensures recovery. 
%To incorporate updated forecasts and mitigate uncertainty, the optimization can also be solved in a receding‑horizon fashion. 
In this approach, we solve the following problem at each time $k = 0$, \dots, $K-1$:
\bneq
\begin{aligned}
&\text{minimize} \quad \sum_{\ell = 1}^{L_k} c(\ell | k) T(\ell | k) \label{pcProb}  \\
&\text{subject to} \\ 
&\ \ y^\text{min}(k + \ell) \leq T(\ell | k) \leq y^\text{max}(k + \ell) \\
&\ \ r^\text{min}(k + \ell) \leq \frac{ T(\ell | k) - T(\ell - 1 | k) }{\Delta t} \leq r^\text{max}(k + \ell)  \\
\end{aligned} 
\eneq
with variables $T(1|k)$, \dots, $T(L_k|k)$. The constraints apply for $\ell = 1$, \dots, $L_k = \min(L, K-k)$, where $L \leq K$ is the prediction horizon. We set $T(0 | k) = T(k)$ and denote by $\ell | k$ an $\ell$-step-ahead forecast or plan made at time $k$. For example, $c(\ell | k)$ is the forecast of $c(k + \ell)$ made at time $k$ using \eqref{price0} or \eqref{price1} with the latest forecasts $\bar \pi(\ell | k)$ and $\bar \pi(\ell-1 | k)$ of $\pi(k+\ell)$ and $\pi(k+\ell-1)$, respectively. Algorithm \ref{mlpc} summarizes this predictive control approach. Problem \eqref{pcProb} is a linear program that off-the-shelf software can solve efficiently and globally. With at most $L$ variables and $4L$ inequality constraints, the worst-case computational complexity scales like $4L^3$ \cite{Boyd2004}. This makes Algorithm \ref{mlpc} tractable even with long forecast horizons and short time steps.

\begin{algorithm}
\caption{Predictive control implementation}
\label{mlpc}
\begin{algorithmic}
\State {\bf Input:} $c(k)$, $y^\text{min}(k)$, $y^\text{max}(k)$, $r^\text{min}(k)$, $r^\text{max}(k)$ for $k = 1$, \dots, $K$

%\State Set $a = \exp(-\Delta t / \tau)$

\For{$k = 0,\dots,K-1$}

    \State Forecast $\bar{\pi}(1 | k)$, \dots, $\bar{\pi}(L_k | k)$
    
    \State Compute $c(1 | k)$, \dots, $c(L_k | k)$ using \eqref{price0} or \eqref{price1}
    
    \State Measure $T(k)$ and set $T(0 | k) = T(k)$

    \State Solve \eqref{pcProb} and set $u(k) = T^\star(1|k)$
    
%    \If{Problem is infeasible}
%        \State Select fallback setpoint $T^\text{fb}(1 | k)$
%        \State Implement $T^\text{fb}(1 | k)$
%    \Else
%        \State Implement optimal setpoint $T^\star(1 | k)$
%    \EndIf

\EndFor
\end{algorithmic}
\end{algorithm}

\subsection{Further Possible Simplifications}

\begin{comment}
{\color{blue} 
Here's where you write up (briefly) the additional assumptions you make. You can repurpose text from the commented-out sections. Please position them as options for people who want to reduce the amount of input data, not as necessary implementation choices. You want to write this section like a menu: Tasty tidbits practitioners can nibble on if they're so inclined. Then in the numerical experiments section, you can say which options you chose from this menu.

For the thing we've been calling zeroth-order -- which is not the only possible implementation that comes from the zeroth-order theory, so I'll call it Algorithm A -- your basic assumption is that $\pi$ is time-invariant. If that's true, then all that matters is the sign of $\pi$ (in Algorithm B, you minimize for heating, maximize for cooling.) You're also choosing to ignore the ramp constraints.

For the thing we've been calling first-order -- again, not the only possible implementation, so I'll call it Algorithm B -- I think you assume that $\eta$ (but not $\tau$) is time-invariant. With that assumption, $\eta$ factors out of the objective function and becomes irrelevant. You don't ignore the ramp constraints for your Algorithm B.

TL;DR: I think the menu items are 1) time-invariant $\eta$, 2) time-invariant $\pi$, 3) time invariant rate limits 4) ignore ramp limits? You assume 2 and 4 for Algorithm A, and for Algorithm B, you assume 1 and 3 only?
}
\end{comment}

Several optional simplifying assumptions may be made to further reduce the input data requirements and commissioning costs of controllers derived from \eqref{reformulatedProblem}.
\subsubsection{Time-invariant efficiency} In some contexts, such as heating with fossil fuels or electric resistance, the parameter $\eta$ is a time-invariant efficiency. In other contexts, such as air conditioning or heating with an electric heat pump, $\eta(k)$ is a COP that generally varies with time due to dependence on dynamic boundary temperatures. If the COP curve is unknown or uncertain, $\eta$ can be treated as time-invariant. %Assuming a time-invariant $\eta$ reflects the strong time-scale separation in the first-order linear model \eqref{1r1c} between indoor air dynamics and the indoor-outdoor temperature difference that largely governs efficiency. 
Under this assumption, $\eta$ factors out of the objective function, so $\eta = 1$ can be assumed without loss of generality.

\subsubsection{Time-invariant price} In some contexts, such as fossil-fueled or resistive heating with constant energy prices, the thermal price $\bar \pi$ is time-invariant. In other contexts, $\bar \pi$ generally varies with time due to a time-varying COP, energy price, or pollutant intensity. Treating $\bar \pi$ as time-invariant removes all dependence on its magnitude, leaving only its sign, so $\bar \pi = \pm 1$ can be assumed without loss of generality (with plus and minus corresponding to heating and cooling, respectively, assuming $\pi^{\text{fuel}} > 0$). This assumption eliminates any need to model or predict $\eta$, $\pi^{\text{fuel}}$, or $\mu$.

\subsubsection{Simplified ramp-rate limits} Assuming time-invariant and/or symmetric ramp-rate limits $r^{\text{min}}$ and $r^{\text{max}}$ reduces the need to model them or the equipment and system parameters on which they depend.

\subsubsection{Neglecting ramp-rate limits} Ramp-rate limits may also be omitted altogether, leaving the low-level controller to enforce equipment capacity constraints.

\section{Numerical Experiments and Discussion}
\label{numerical}
\subsection{Setting, Building, and Input Data}
\label{testbed}

\begin{comment}
    
We evaluate the nearly commissioning‑free controllers from Sections~\ref{problem0} and \ref{problem1} in a simulated deployment aimed at minimizing heating energy cost in a single‑story detached residential building located in West Lafayette, Indiana, USA during the winter of 2022.

To provide a more realistic testbed than the first‑order model used for controller design, we simulate the building using a second-order thermal model with two thermal masses. The ``shallow'' mass represents the indoor air and tightly coupled material, such as ducts or the top layers of interior surfaces. The ``deep'' mass represents heavier construction materials in walls and floors. The dynamics are 
\[
\begin{aligned}
C \dot T(t) &= \frac{\theta(t) - T(t)}{R} + \frac{T_m(t) - T(t)}{R_m} + q(t) + q^\text{exog}(t) \\
C_m \dot T_m(t) &= \frac{T(t) - T_m(t)}{R_m} .
\end{aligned}
\]
Here $T(t)$ is the lumped temperature of the indoor air and shallow thermal mass;  $T_m(t)$ is the deep mass temperature. The parameter values are $C = 2.16$~kWh/$^\circ$C, $C_m = 22.68$~kWh/$^\circ$C, $R = 3.27$~$^\circ$C/kW, and $R_m = 0.625$~$^\circ$C/kW. The outdoor temperature $\theta(t)$ is taken from West Lafayette weather data. Exogenous heat gains are modeled as $q^{\text{exog}}(t) \sim \mathcal{N}(2,\,1/36)$~kW, independently and identically distributed. We discretize the dynamics using the matrix exponential with a five-minute time step and a zero-order hold on the input signals.

\end{comment}

We evaluate the nearly commissioning‑free controllers from Sections~\ref{problem0} and \ref{problem1} in a simulated deployment aimed at minimizing cooling energy cost during a 2026 summer week in the DC Nanogrid House, a two‑story detached residential building located in West Lafayette, Indiana \cite{Pergantis-etal-2024a}.

To provide a realistic testbed, we simulate the building using a data-driven second-order thermal model with higher fidelity than the first-order model that motivated control design. The model distinguishes solar disturbance inputs from other exogenous loads and includes two thermal masses. The ``shallow'' mass represents the indoor air and tightly coupled material such as ducts or the top layers of interior surfaces, while the ``deep'' mass represents heavier construction materials in walls and floors. The dynamics are 
\[
\begin{aligned}
C \dot T(t) &= \frac{\theta(t) - T(t)}{R} + \frac{T_m(t) - T(t)}{R_m} + q(t) + q^\text{exog}(t) \\
C_m \dot T_m(t) &= \frac{T(t) - T_m(t)}{R_m},\qquad q^\text{exog}(t) = \alpha q^\text{sol}(t) + \beta.
\end{aligned}
\]
Here $T$, $\theta$, $q$, $q^\text{exog}$, $R$, and $C$ have the same interpretations as in \eqref{continuousDynamics}, $T_m$ ($^\circ$C) is the deep-mass temperature, $C_m$ (kWh/$^\circ$C) is the deep-mass capacitance, $R_m$ ($^\circ$C/kW) is the resistance between $T$ and $T_m$, $q^\text{sol}$ (kW/m$^2$) is the global solar irradiance on a horizontal surface, $\alpha$ (m$^2$) is the effective solar absorption area, and $\beta$ (kW) is a constant exogenous thermal power. We discretize the second-order model using the matrix exponential with a five-minute time step and a zero-order hold on all input signals. 

We fit the parameter values $C = 7.08$~kWh/$^\circ$C, $C_m = 17.54$~kWh/$^\circ$C, $R = 2.15$~$^\circ$C/kW, $R_m = 0.504$~$^\circ$C/kW, $\alpha = 3.09$~$\text{m}^2$, and $\beta = 2.48$~kW to measured data from summer in West Lafayette (one training week and one validation week). With the indoor temperature reinitialized at the measured value every 24 hours, the model predicts the indoor temperature with a root-mean-square error of $0.43$~$^\circ$C over a third, held-out week used in the simulations below.

\begin{comment}
The building is conditioned by a reversible electric heat pump with
input power constraints $p \in [0,\,7.24]$~kW sized to meet the design heating load. The coefficient of performance is linear in the outdoor temperature: $\eta(t)= 1.5 + 0.057 (\theta(t) - 7)$. The heat pump's low-level control system perfectly tracks indoor temperature setpoints up to saturation to enforce equipment capacity constraints.

\end{comment}

The building is conditioned by a reversible heat pump with
4.5 kW of input electric power capacity and a COP quadratic in the outdoor temperature: $\eta = 0.0048 \theta^2 - 0.4551 \theta + 13.549$ \cite{Pergantis-etal-2024a}. We assume that the heat pump's low-level control system perfectly tracks indoor temperature setpoints up to saturation at equipment capacity limits.

We consider two electricity pricing structures representative of those serving the vast majority of residential and commercial buildings in the United States, based on rate schedules from a large, regional utility \cite{duke_tariff}. Under the flat-rate structure, electricity is priced at $\pi^{\text{fuel}}= 0.187 \text{ \$/kWh}$.
Under the time-of-use (TOU) structure, electricity prices vary with anticipated grid demand:

\[
\pi^\text{fuel}(t) =
\begin{cases}
0.086~\$/\text{kWh}, & 12~\text{AM} \le t < 4~\text{AM} \\
0.214~\$/\text{kWh}, & 5~\text{PM} \le t < 9~\text{PM} \\
0.143~\$/\text{kWh}, & \text{otherwise}.
\end{cases}
\]
\begin{comment}
\[
\pi^\text{fuel}(t)=
\begin{cases}
0.086~\$/\text{kWh}, & 12~\text{AM} \le t < 4~\text{AM} \\
0.211~\$/\text{kWh}, & 6~\text{AM} \le t < 8~\text{AM} \\
0.211~\$/\text{kWh}, & 5~\text{PM} \le t < 9~\text{PM} \\
0.143~\$/\text{kWh}, & \text{otherwise}.
\end{cases}
\]
\end{comment}

This combination of realistic thermal dynamics, equipment behavior, and pricing structures creates a representative setting for assessing controller performance.

\subsection{Control Performance }
\label{compare}
We evaluate the zeroth‑order and first‑order controllers under both flat‑rate and TOU electricity pricing. Each controller operates in closed loop using Algorithm \ref{mlpc} with a 24‑hour prediction horizon. Temperature limits follow a typical residential occupancy schedule: 
\[
(y^\text{min}(t),y^\text{max}(t)) = 
\begin{cases}
(16, 24) \text{ $^\circ$C,}& 9\text{ AM} \leq t \leq 4\text{ PM} \\
(19, 21) \text{ $^\circ$C,}& \text{otherwise.}
\end{cases}
\]

For both controllers, we assume symmetric, time‑invariant ramp‑rate limits $r^{\text{max}} = -r^{\text{min}} = r$. The zeroth-order controller uses a negative time-invariant thermal price $\bar \pi=-1$~$\$$/kWh, while the first-order controller assumes a time‑invariant efficiency $\eta$ such that $\bar{\pi} = -\pi^\text{fuel}$. Where applicable, controllers use parameter estimates $r_{\text{est}} = 2$~$^\circ$C/h, obtained from a simulated step-input test, and $\tau_{\text{est}} = 8.96$~h, derived from averages across more than 2,000 homes in a comparable climate region and season \cite{John2021}. Under these assumptions, the deployed controllers require, at most, forecasts of energy prices and occupant temperature preferences over a prediction horizon, an indoor temperature measurement, and two constant parameters which are easily estimated from historical data, simple tests, or rules of thumb.

\begin{comment}
For the zeroth-order controller, we neglect ramp-rate constraints and assume $\bar \pi$ is time-invariant. For the first‑order controller, we assume symmetric, time‑invariant ramp‑rate constraints $r^{\text{max}} = -r^{\text{min}} = r$ and assume a time‑invariant $\eta$ such that $\bar{\pi} = \pi^\text{fuel}$. The first-order controller receives price forecasts and uses mildly inaccurate parameter estimates $\tau_{\text{est}} = 6$~h and $r_{\text{est}} = 4$~$^\circ$C/h that are attainable through intuition, a short test, or rules of thumb. The true building exhibits a dominant time constant of $\tau = 7$~h and time-varying ramp-rate limits with average values $\bar r^{\min} = -2.60$~$^\circ$C/h and $\bar r^{\max} = 4.35$~$^\circ$C/h derived from the second-order model.
\end{comment}

Two benchmark controllers define the performance envelope. The baseline thermostat maintains a constant setpoint at the least-cost temperature within the occupied comfort range, 21~$^\circ$C. The omniscient optimal controller has perfect knowledge of the full second-order model structure and parameters as well as perfect forecasts of all disturbances and prices. It provides a fundamental performance limit by solving the optimal control problem exactly over the full simulation horizon. The zeroth-order, first-order, and omniscient optimal control problems are solved via linear programming using MATLAB’s \texttt{linprog} function~\cite{matlab2024}.

Figs.~\ref{comp_flat} and \ref{comp_tou} show timeseries data under flat-rate and TOU pricing, respectively, for two days with the baseline, zeroth-order, first-order, and omniscient controllers. Subplots show the electricity price (top, left axis), outdoor temperature (top, right axis), cooling supplied (middle), and indoor temperature (bottom). Table~\ref{results} summarizes the resulting electricity costs, savings, and cumulative temperature constraint violations (measured in degree-hours, $^\circ$Ch) over one week. These are stylized simulation results; ultimately, performance should be evaluated through deployment in a real building.

\begin{figure}[t]
\begin{center}
\includegraphics[width=0.49\textwidth]{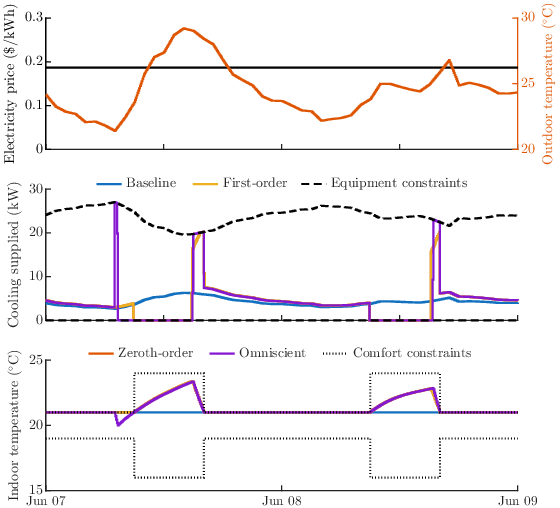}
\end{center}
\caption{Temperature and cooling trajectories under flat-rate pricing. The zeroth- and first-order controllers behave similarly to each other, but miss the omniscient optimal controller's load-shifting event on June 7.}
\label{comp_flat}
\end{figure}

\begin{figure}[t]
\begin{center}
\includegraphics[width=0.49\textwidth]{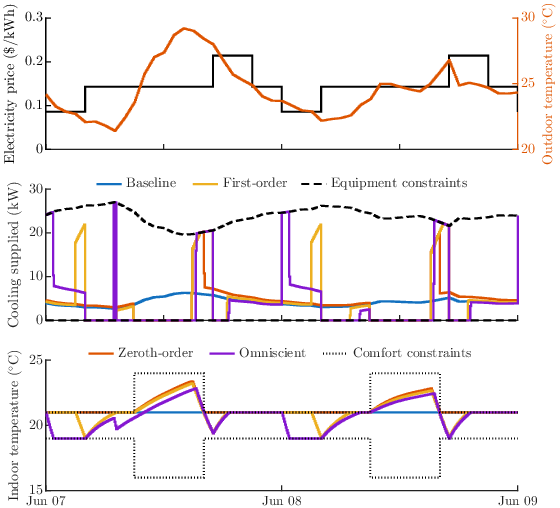}
\end{center}
\caption{Temperature and cooling trajectories under TOU pricing. The zeroth-order controller does no price-based load shifting. The first-order controller does, but less aggressively than the omniscient optimal controller.}
\label{comp_tou}
\end{figure}

Table \ref{results} shows that under flat-rate pricing (Fig.~\ref{comp_flat}), the zeroth-order, first-order, and omniscient controllers all reduce the electricity cost by 3.9--4.5\% relative to the baseline thermostat. Most savings arise from relaxing temperature bounds during unoccupied periods; the omniscient controller's load-shifting events (e.g., bottom plot, purple curve, first day) in response to outdoor-temperature trends contributes negligible savings in this scenario. The zeroth- and first-order controllers exhibit minor comfort violations during transitions from unoccupied to occupied hours due to imperfect ramp-rate estimates, but these deviations are small and could be mitigated through conservative occupancy schedules or adaptive ramp-rate estimation. Overall, the results indicate that simple unoccupied‑setback strategies with smart recovery, already available in many programmable thermostats, capture nearly all economically meaningful savings for flat-rate pricing structures, which serve roughly 90\% of United States residential and commercial customers \cite{forrester2024}.

\begin{comment}
The zeroth-order controller produces brief comfort violations during transitions from unoccupied to occupied hours due to our choice to neglect the ramp-rate constraints, but this effect is small and can be mitigated by assigning conservative occupancy schedules or specifying a conservative ramp-rate constraint.
\end{comment}

\begin{table}[t]
\caption{Simulation results for June 4--10, 2026}
\begin{center}
\renewcommand{\arraystretch}{1.1}
\setlength{\tabcolsep}{3pt}
\begin{tabular}{c|c|c|c|c}
 & & Zeroth-order& First-order& Omniscient \\
\hline

\multirow{4}{*}{\rotatebox{90}{Flat-rate}}& Cost (\$)& 24.48& 24.48& 24.63\\
 & Savings vs. baseline (\$)& 1.15& 1.15& 1.00\\
 & \% of omniscient savings& 115.7& 115.7& 100\\
 & Discomfort ($^\circ$Ch)& 0.16& 0.16& 0 \\
\hline

\multirow{4}{*}{\rotatebox{90}{TOU}}& Cost (\$)
& 20.44& 18.96& 18.65\\
 & Savings vs. baseline (\$)
& 0.56& 2.03& 2.35\\
 & \% of omniscient savings
& 23.9& 86.6& 100\\
 & Discomfort ($^\circ$Ch)& 0.16& 0.07& 0 \\
\end{tabular}
\normalsize
\end{center}
\label{results}
\end{table}
{
\begin{comment}
\begin{table}[t]
\caption{Simulation results for January 2--8, 2022}
\begin{center}
\renewcommand{\arraystretch}{1.1}
\setlength{\tabcolsep}{3pt}
\begin{tabular}{c|c|c|c|c}
 & & Zeroth-order& First-order& Omniscient \\
\hline

\multirow{4}{*}{\rotatebox{90}{Flat-rate}}& Cost (\$)& 60.05& 60.25& 60.16\\
 & Savings vs. baseline (\$)& 3.71& 3.51& 3.60\\
 & \% of omniscient savings& 103.1& 97.5& 100\\
 & Discomfort ($^\circ$Ch)& 7.38& 0.07& -- \\
\hline

\multirow{4}{*}{\rotatebox{90}{TOU}}& Cost (\$)
& 70.91& 69.38& 67.09\\
 & Savings vs. baseline (\$)
& 2.87& 4.40& 6.69\\
 & \% of omniscient savings
& 42.9& 65.8& 100\\
 & Discomfort ($^\circ$Ch)& 7.38& 0.07& -- \\
\end{tabular}
\normalsize
\end{center}
\label{results}
\end{table}
\end{comment}
}

\normalsize
Table \ref{results} shows modest savings for the zeroth-order controller under TOU pricing (Fig.~\ref{comp_tou}); however, performance is highly sensitive to the alignment between pricing periods and occupancy transitions. Because the zeroth-order controller lacks dynamics and price awareness, it can increase cost relative to baseline if the energy-intensive unoccupied-setback recovery coincides with peak-price windows. This behavior highlights the limitations of trivially simple controllers under increasingly complex pricing structures. In contrast, the first-order and omniscient controllers shift load to avoid peak-price periods and generate significant savings. The first-order controller improves performance relative to the zeroth-order controller by precooling the air immediately before price spikes. The omniscient controller---which has access to the full second-order dynamics---saves more money through longer precooling actions that store energy in the deep thermal mass. Nevertheless, the first-order controller captures 86.6\% of the omniscient optimal savings in this example.

\subsection{Sensitivity to Parameter Estimation}
\label{sensitivity}

The first-order controller requires parameter estimates $\tau_{\text{est}}$ and $r_{\text{est}}$. To assess performance sensitivity to these parameter values, we simulate one week of deployment under the TOU pricing structure for 2,500 combinations of $\tau_{\text{est}}$ and $r_{\text{est}}$ spanning physically plausible ranges for typical residential buildings. Each controller instance is evaluated by its total electricity-cost savings relative to the omniscient controller. %Fig.~\ref{sensitivity} summarizes the results.

Fig.~\ref{sensitivityFig} shows that the first-order controller's cost savings are largely insensitive to variation in $\tau_{\text{est}}$. This reflects the structure of TOU pricing: large price differentials make load-shifting almost always beneficial, so the precise value of $\tau$ has limited influence on the qualitative control action. In contrast, savings and discomfort are more sensitive to $r_{\text{est}}$, which limits precooling and setback-recovery rates; however, costs never increased relative to the baseline thermostat, nor did discomfort exceed 5.7~$^\circ$Ch for any parameter combination over the full simulation week.% This one-sided risk profile may make the controller attractive in commissioning-constrained settings.

\begin{figure}[t]
\begin{center}
\includegraphics[width=0.49\textwidth]{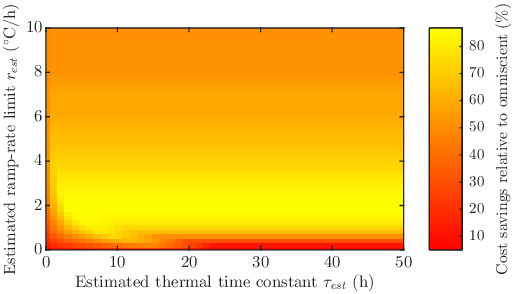}
\end{center}
\caption{The first‑order controller's week-long TOU cost savings (color shading) remain high relative to the omniscient controller despite wide variation in the estimated time constant $\tau_\text{est}$ and ramp-rate limit $r_\text{est}$.%This suggests that parameter values can be selected through engineering judgment or default settings rather than hand-tuned.
}
\label{sensitivityFig}
\end{figure}

These results suggest that the controller may perform very well with relatively coarse choices of $\tau_\text{est}$ and $r_\text{est}$, potentially limiting the need for detailed building-specific calibration and offering opportunities to further reduce commissioning burden. However, a more comprehensive evaluation across building types, climates, and pricing structures is needed to determine the generality of this finding.

\begin{comment}
These results suggest that $\tau_\text{est}$ and $r_\text{est}$ do not require precise tuning. Either parameter could plausibly be instantiated using factory default values or simple heuristics with little to no building-specific calibration. This would reduce the thermal-parameter estimation requirement to essentially zero and lower commissioning complexity to that of a standard programmable thermostat. Further work is needed to determine whether these robustness properties generalize across building types, climates, and pricing structures.
\end{comment}

\section{Summary, Limitations, and Future Work}
\label{conclusions}
This paper addressed barriers to the adoption of advanced supervisory HVAC control by introducing two controllers that substantially reduce commissioning effort relative to RL and conventional MPC. In simulations, both controllers delivered near‑perfect performance under the flat‑rate electricity pricing structures that serve the vast majority of United States buildings, using actions similar to those already implemented in many thermostats. Although these existing strategies proved insufficient for increasingly dynamic pricing structures \cite{forrester2024}, the pricing‑aware first‑order controller captured a large fraction of the attainable savings with only modest parameter requirements. Taken together, these results suggest that for supervisory HVAC control with linear objectives, simple methods may recover most of the value of advanced control while avoiding the modeling, data, and tuning burdens that hinder real‑world deployment.

\begin{comment}
This paper addressed barriers to adoption of advanced supervisory HVAC control, introducing two controllers that minimize commissioning effort by requiring only occupant comfort preferences and pricing structures. In simulations, both controllers achieve near-optimal costs under the flat-rate electricity pricing structures that serve 90\% of residential and commercial customers in the United States. For more dynamic pricing structures that are increasing in prevalence \cite{forrester2024}, the first-order controller could likely capture much of the attainable cost savings. These results suggest that RL and conventional MPC may be unnecessary for HVAC control applications with linear objectives.
\end{comment}

Several limitations merit further investigation. Real buildings often exhibit multi-zone interactions, equipment and thermal nonlinearities, imperfect low-level control, and uncertain occupant behavior not captured in this study. Future work could evaluate the proposed controllers in a broader range of building types, climates, occupancy schedules, and pricing structures, comparing them to established advanced controllers and simple heuristics in simulation or field experiments. Additionally, extending the framework in this paper to incorporate higher-order thermal dynamics could enable engagement of deep thermal mass, potentially unlocking additional savings in scenarios with large price variations.

\section{Acknowledgments}

We gratefully acknowledge support for WGD from the Purdue Mechanical Engineering Ingersoll-Rand Fellowship, for AJK and LDRP from the ASHRAE Grant-in-Aid Award, and for LDRP from the NSF Graduate Research Fellowship. WGD used Microsoft Copilot \cite{copilot} for syntax and grammar refinement and figure formatting. All AI-generated code and text were reviewed by WGD, who takes full responsibility for the accuracy and integrity of the final work.

\bibliography{IEEEabrv,ltc}

\end{document}